\newif\ifdraft
\draftfalse 

\newif\ifarxiv
\arxivtrue

\ifdraft
\documentclass[draft,12pt,a4paper]{article}
\else
\documentclass[12pt,a4paper]{article}
\fi
\long\def\comment#1{}
\usepackage[utf8]{inputenc} 
\usepackage[T1]{fontenc}

\usepackage{amsmath, amsthm, amssymb}
\usepackage{bm}
\usepackage{booktabs} 
\usepackage{currfile} 
\usepackage[final]{listings}
\usepackage[final]{graphicx}
\usepackage{enumitem} 
\usepackage{multirow}
\usepackage{txfonts} 
\usepackage{tikz} 
\usetikzlibrary{calc,positioning,arrows,shapes}
\usepackage[obeyDraft,textsize=tiny]{todonotes}
\ifarxiv
\usepackage[natbib=true,style=authoryear,backend=bibtex,maxbibnames=100,maxcitenames=3,firstinits=true,isbn=false]{biblatex} 
\else
\usepackage[natbib=true,style=authoryear,backend=bibtex,maxbibnames=100,maxcitenames=3,giveninits=true,isbn=false]{biblatex} 
\fi
\renewbibmacro{in:}{%
  \ifentrytype{article}{}{\printtext{\bibstring{in}\intitlepunct}}}

\ifarxiv
\else
\fi

\usepackage{\ifarxiv\else ../../../../lib/tex/clange/\fi unicode-chars}
\usepackage{\ifarxiv\else ../../../../lib/tex/clange/\fi pmisabelle}

\usepackage[babel]{csquotes} 
\MakeAutoQuote{“}{”}
\MakeAutoQuote*{‘}{’}

\ifarxiv
\newcommand*{\href}[2]{#2}
\else
\usepackage{hyperref} 
\fi

\newtheorem{theorem}{Theorem}

\begin{document}

\title{An Introduction to Mechanized Reasoning\footnote{We are grateful to Makarius Wenzel for help refining our code, to Marco Caminati for research assistance, to Peter Cramton, Paul Klemperer, Peter Postl, Indra Ray, Rajiv Sarin, Arunava Sen and Ron Smith for comments, and to the EPSRC for funding (EP/J007498/1). Rowat thanks Birkbeck for its hospitality.  The presentation of the formal proof of Vickrey's theorem is based on \citet{ke-la-ro-14}.  Finally, we are grateful to two anonymous referees and the co-editor for working with us to improve this paper.}}

\author{Manfred Kerber\footnote{School of Computer Science, University of Birmingham, UK}
\and Christoph Lange\footnote{Fraunhofer IAIS and University of Bonn, Germany}
\and Colin Rowat\footnote{Department of Economics, University of Birmingham, Edgbaston B15 2TT, UK, c.rowat@bham.ac.uk, +44 121 414 3754.  Corresponding author}
}

\date{\today}

\maketitle

\begin{abstract}
  Mechanized reasoning uses computers to verify proofs and to help discover new theorems.  Computer scientists have applied mechanized reasoning to economic problems but -- to date -- this work has not yet been properly presented in economics journals.  We introduce mechanized reasoning to economists in three ways.  First, we introduce mechanized reasoning in general, describing both the techniques and their successful applications.  Second, we explain how mechanized reasoning has been applied to economic problems, concentrating on the two domains that have attracted the most attention: social choice theory and auction theory.  Finally, we present a detailed example of mechanized reasoning in practice by means of a proof of Vickrey's familiar theorem on second-price auctions.

  \vspace{0.5cm}


  \vspace{0.5cm}


\ifdraft
  \vspace{0.5cm}

  {\it This file}: \currfilepath
\fi
\end{abstract}

\section{Introduction}

Mechanized reasoners automate logical operations, extending the scope of mechanical support for human reasoning beyond numerical computations (such as those carried out by a calculator) and symbolic calculations (such as those carried out by a computer algebra system).  Such reasoners may be used to formulate new conjectures, check existing proofs, formally encode knowledge, or even prove new results.  The idea of mechanizing reasoning dates back at least to \citet{Leibniz86}, who envisaged a machine which could compute the validity of arguments and the truth of mathematical statements.  The development of formal logic from 1850 to 1930, the advent of the computer, and the inception of \textit{artificial intelligence} (AI) as a research field at the Dartmouth Workshop in 1956 all paved the way for the first mechanized reasoners in the 1950s and 1960s.\footnote{Perhaps unsurprisingly, \citeauthor{gar-52} was ahead of his time in mechanized reasoning as well: four years before his regular columns with \textit{Scientific American} began, his first article for them included a template allowing readers to make their own mechanized reasoners -- out of paper.}

Since then, mechanized reasoning has been both less and more successful than anticipated.  In pure maths, mechanized reasoning has helped prove only a few high-profile theorems.  Perhaps surprisingly -- although consistent with the greater success of applied AI over `pure' AI -- mechanized reasoning and formal methods\todo{CL: really “from the axioms ZFC”, or “from the axioms \emph{of} ZFC”?  CR@CL: yes}\footnote{The term \textit{formal methods} is used here to denote approaches to establishing the correctness of mathematical statements to a precision that they can be meticulously checked by a computer.  Rather than being seen as distinct from other mathematical methods, researchers in the area see them as the next step in mathematics' march towards greater precision and rigor \citep{wie-08}.  Consider: ``A Mathematical proof is rigorous when it is (or could be) written out in the first-order predicate language $L \left( \in \right)$ as a sequence of inferences from the axioms ZFC'' \citep{mac-86}.  The advantages of taking this next step with computers include: a computer system is never tired or intimidated by authority, it does not make hidden assumptions, and can easily be rerun.  A pioneer of mechanized reasoning -- who saw himself building on Bourbaki's formalism -- referred to computers as ``slaves which are such persistent plodders'' \citep{wan-60}.} have enjoyed greater success in industrial applications, as applied to both hardware and software design. In the past decade or so, computer scientists have also begun to apply formal methods to economics.

A central inspiration for this recent work are \citeauthor{gea-05}' three brief proofs of Arrow's impossibility theorem \citep{gea-05}.\footnote{All three use \citeauthor{bar-80}'s replacement of Arrow's \textit{decisive voter} with a \textit{pivotal voter} \citep{bar-80}.  \citet{bar-83} also used this approach to find a direct proof of the Gibbard-Sattherthwaite theorem.}  Initially, \citet{nip-09}, \citet{wie-07}, and \citet{wie-09} used theorem provers to encode and verify two of \citeauthor{gea-05}' proofs. A subsequent generation of work, drawing on the inductive proof of Arrow's theorem in \citet{suz-00}, used formal methods to discover new theorems.  \citet{ta-li-09} introduced a hybrid technique, using computational exhaustion to show that Arrow holds on a small base case of two agents and three alternatives, and then manual induction to extend that to the full theorem.  By inspecting the results of the computational step, they were able to discover a new theorem subsuming Arrow's.  \citet{ta-li-11} used this approach -- exhaustively generating and evaluating base cases, and then using a manual induction proof to generalize the results -- to establish uniqueness conditions for pure strategy Nash equilibrium payoffs in two player static games; they published manual proofs of two of the most significant theorems discovered this way in \citet{ta-li-11geb}.  \citet{ge-en-11} used the approach to generate 84 impossibility theorems in the `ranking sets of objects' problem \citep{ba-bo-pa-04}.

To date, the economics literature remains almost untouched by research applying mechanized reasoning to economic problems.\footnote{A recent symposium on economics and computer science, involving central figures at the interface between the disciplines, made no mention of mechanized reasoning \citep[q.v.][]{bl-ea-kl-kl-ta-15}.}  The one exception that we are aware of is \citet{ta-li-11geb}, whose two theorems were discovered computationally, but proved manually.\footnote{The process by which the theorems were discovered is described in \citet{ta-li-11}; \citet{ta-li-11geb} itself is all but silent on its mechanized origins.}  As it is our view that these tools will become increasingly capable, this paper aims to introduce economists to mechanized reasoning.\footnote{For more general introductions, see \citet{wie-08} and \citet{av-ha-14}.  \citet{harrison-ab} introduces mechanized reasoning alongside computer algebra, presenting something of a unified view.}  It does so by means of three analytical lenses, each with narrower scope but greater magnification than its predecessor.

First, Section \ref{se:imr} presents an overview of mechanized reasoning in general.  We do so by setting out a classificatory scheme, with the caveat that it should not be seen as implying a partition on the field: interesting research will straddle boundaries, perhaps even forcing them to be redefined.\footnote{For example, we shall see that mechanized theorem discovery is usually associated with inductive reasoning.  However -- in economic examples -- the most fruitful examples of theorem discovery \citep{ta-li-09,ta-li-11,ta-li-11geb,ge-en-11} have combined very simple deductive reasoning systems with human intelligence.}

Second, Section \ref{se:amre} surveys the emerging literature applying mechanized reasoning to economics.  We structure this survey primarily according to the problem domain within economics, referring only secondarily to our classificatory scheme.  We do this to focus on the economic insights -- primarily within social choice and auction theory -- made possible by these techniques, rather than on the techniques \textit{per se}.

Finally, to make this introduction more concrete, Section \ref{se:mpv} provides an example of what mechanized reasoning looks like in practice, presenting a blueprint of a mechanized proof of Vickrey's theorem on second-price auctions.  We present such an established theorem to focus attention on its implementation.

Section \ref{se:concle} concludes, and suggests some possible next steps for mechanized reasoning in economics.

\section{Mechanized reasoning} \label{se:imr}

Our overview of mechanized reasoning distinguishes between deductive and inductive systems.  While the distinction has been recognized at least since Aristotle, deductive reasoning -- which allows reliable inference of unknown facts from established facts -- has been in the focus of the mechanized reasoning community.  Inductive reasoning also generalizes from individual cases, but does not restrict itself to reliable inferences; the cost of this additional freedom is that its conjectures must then be tested.

\subsection{Deductive reasoning}\label{subsection:deductive-reasoning}

Historically, deductive reasoning systems were among the first AI systems, dating back to the 1950s.  While the origins of deductive reasoning date to at least Aristotle, modern advances in this area built on the work of logicians in the second half of the 19th century and the start of the 20th \citep[e.g.][]{WhiteheadRussell10}.  At the Dartmouth Workshop in 1956, \citeauthor{ne-si-56} introduced the \href{https://en.wikipedia.org/wiki/Logic_Theorist}{Logic Theorist}, an automated reasoner which re-proved 38 of the 52 theorems in Whitehead and Russell's \textit{Principia Mathematica}~\citep{WhiteheadRussell10}.\footnote{According to \citet{mcc-04}, Russell himself ``responded with delight'' when shown the Logic Theorist's proof of the isosceles triangle theorem, whose proof was more elegant than their manual one.}

Abstractly, a deductive reasoner implements a \textit{logic} -- which is comprised of a \textit{syntax} defining well-formed formulae and a \textit{semantics} assigning meaning to formulae -- and a \textit{calculus} for deriving formulae (called theorems) from formulae (called premises or axioms).  Historically, subfields of mechanized reasoning have been defined by choice of logic, calculus and problem domain.  This section provides a classificatory scheme based, first, on the choice of calculus.  Following the choice of calculus, a logic is chosen to balance expressiveness and tractability. Finally, the problem domain itself will dictate some of the specialized features of a mechanized reasoner.

When a mechanized reasoner applies the calculus' permissible operations to the axioms to obtain new, syntactically-correct formulae it does not make use of the semantics: the semantics, or ascribed meanings, yield models that may assist human intuition, but which are not necessary to the formal process of reasoning itself.\footnote{Beginning with Euclid's efforts to axiomatize geometry, logicians have produced syntactical descriptions that make semantic references  obsolete: Hilbert allegedly said that we would still have an axiomatization of geometry if we replaced the words `point', `line', and `plane' by `beer mug', `bench', and `table' \citep[p.6]{Hoffmann13}.}  Crucially, mechanized reasoning involves manipulating symbols.\footnote{That this was an insight at one point may be inferred from \citeauthor{tur-36}'s famous explanation that, ``computing is normally done by writing certain symbols on paper'' \citep{tur-36}}

Thus, mechanized deductive reasoning since the Logic Theorist has seen reasoning as a search task for a syntactically well-defined goal.\footnote{As noted by \citet{harrison-ab}, specialist provers have also been developed for particular problems for which more structured approaches than general search are appropriate.}  Further, as the spaces through which search occurred was potentially large, successful reasoning would use \emph{heuristics} to avoid unprofitable sequences of operations.  From this point of view, mechanized reasoning operates as chess computers do.\footnote{Indeed, Newell's collaboration with Simon began after the latter became aware of the former's work on a chess machine.}  For a chess computer, the premises' intended semantic interpretations are the board, its pieces and their positions; the calculus specified permissible moves.  A chess computer could then test manually discovered solutions to chess puzzles by verifying that each move satisfies the requirements of its calculus, with the final operation yielding the goal-formula.  More ambitiously, and interestingly, chess programs discover solutions (e.g.\ sequences of winning moves) by searching through permissible operations, with the benefit of heuristics (e.g.\ regarding relative values of pieces).


\comment{Let $\Gamma$ denote a set of premises, and $\varphi$ a formula; then $\Gamma \models \varphi$ means ``$\varphi$ follows from $\Gamma$'' (semantically, if all the formulae in $\Gamma$ always hold in the logical environment, then $\varphi$ holds in it as well) while $\Gamma \vdash \varphi$ means ``$\varphi$ can be derived from $\Gamma$'' (by manipulating symbols $\varphi$ can be generated from $\Gamma$ by applying the calculus' rules).  A calculus is \textit{sound} if $\left( \Gamma \vdash \varphi \right)$ implies $\left( \Gamma \models \varphi \right)$.  Deductive reasoning is sound; inductive reasoning, considered below, is not.}

A set of premises and a formula may be related in two different ways.  First, the \textit{semantic consequence relation} describes situations in which the formula \textit{follows from} the premises: if the symbols in the premises are interpreted in such a way that the formulae in the premises are all true, then the formula is also true when the symbols in it are interpreted in the same way. Second, the \textit{syntactic derivibility relation} describes situations in which the formula \textit{can be derived from} the premises: it is possible to generate the formula from the premises by applying a fixed set of so-called calculus rules. (An example of such a rule is \textit{modus ponens}: From $A$ and $A \rightarrow B$ it is possible to derive $B$, where $A$ and $B$ may match any formal expression).  A proof that applies such rules, without any appeals to intuition or to the reader filling in steps on her own, is called a \textit{formal proof} of the formula using the premises.

A calculus is called \textit{sound} if only formulae can be derived from the premises that actually follow from them.  Deductive reasoning is sound; inductive reasoning, considered below, is not.

A calculus is \textit{complete} if it allows derivation of any formula that follows from the set of premises.  A calculus is \textit{decidable} if, for any set of premises and any formula, there is a procedure that either derives the formula from the premises or proves that no such derivation exists; a calculus is \textit{semi-decidable} if a procedure exists that derives the formula from premises, whenever the formula follows from them (but may not terminate if it does not).

Decidability typically depends on the expressiveness of the logic used: more expressive logics model a richer set of concepts, but are generally harder to manipulate.  While ambitious exercises in mechanized reasoning often begin by specifying a suitably tailored logic\footnote{See, for example, the \textit{judgement aggregation logic} (JAL) of \citet{ag-ho-wo-11}.}, we largely restrict our attention to some of the best known \textit{classical logics}.\footnote{The 17 volumes in the second edition of \citet{gab-gue} make clear that the classical logics are a small subset of all logics.}

\paragraph{Propositional (Boolean) logic:}
\textit{Propositional} or \textit{Boolean logic}, the simplest classical logic, only uses propositional variables -- which are either true or false -- and \textit{connectives} such as $\wedge$ (and), $\vee$ (or), $\neg$ (not), and $\rightarrow$ (implies).  An example of a propositional formula is
\[
  \mathit{first\_bidder\_bids\_highest} \wedge \mathit{second\_bidder\_bids\_lowest}.
\]
Propositional logic can only make concrete, finite statements, but has a sound, complete and decidable calculus.

An advantage of this decidability is that it may allow \textit{push-button} technology, which does not require specialist knowledge in order to use. Once a problem is adequately represented a corresponding system solves the problem fully automatically.

\paragraph{First-order logic:}
\textit{First-order logic} (FOL) is more expressive.  First, it can speak about  objects (e.g.\ ``bidder $b_1$'') 
and their properties (e.g.\ ``bidder $b_1$ wins auction'', $\mathit{bidder}(b_1) \land \mathit{wins} \left( b_1 \right)$). Second, $\exists$ and $\forall$ allow quantification over objects.  For example, ``every losing bidder pays nothing'' may be expressed as
\begin{equation} \label{eq:FOL-bid}
  \forall i\,.\, \mathit{bidder}(i) \rightarrow(\neg\mathit{wins} \left( i \right) \rightarrow \mathit{pay} \left( i \right) = 0).
\end{equation}
Expressions like $\mathit{wins}$ are called \textit{predicates}, Boolean functions which -- when applied to their arguments -- evaluate to either true or false.  G\"{o}del's completeness theorem proves that FOL has a sound and complete calculus, but FOL has only semi-decidable calculi. Furthermore, FOL is not expressive enough to express the finitude or (per negation) infinity of the non-empty sets of objects.\footnote{Thus, FOL could not express that only finitely many bidders participate in an auction.}

\textit{Many-sorted FOL} uses \textit{sorts} to extend first-order logic, not to add to its expressiveness, but to allow more concise representations, and -- therefore -- more efficient proving.  Sorts restrict the instantiation of variables to expressions of a certain sort. For instance, sorts allow us to specify that variable $i$ is a bidder, and variable $x$ a good.
Formula (\ref{eq:FOL-bid}) is then more precisely stated as:
\begin{equation} \label{eq:FOL-bid-sorted}
  \forall i_{\mathit{bidder}}\,.\,\neg\mathit{wins} \left( i \right) \rightarrow \mathit{pay} \left( i \right) = 0.
\end{equation}
$i$ (with the sort $\mathit{bidder}$ mentioned only at the first occurrence) can be instantiated now by terms of sort bidder, but not by those of sort good, thus reducing the search space for a proof.
Sorted formulae can be translated to unsorted formulae by converting the sorts to unary predicates (which take a single argument).

\paragraph{Higher-order logic:}
\textit{Higher-order logic} (HOL) enriches the expressiveness of FOL by extending quantification to predicates and functions.  It also allows predicates and functions to take certain\footnote{Unrestricted formula building leads to antinomies as discovered by Russell. The introduction of types imposes a hierarchy on logical objects, including predicates.  This disables circular constructs such as $X(Y) := \neg Y(Y)$, which -- when $Y$ is instantiated with $X$ -- produces the set of all $X$ for which $X \not \in X$, Russell's famous antinomy.} other predicates and functions as arguments.  For example, bids, $b$, are both a function from bidders to prices and an argument (along with $N, v$ and $A$) in the predicate
\[
  \mathit{equilibrium\_weakly\_dominant\_strategy} \ N \ v \ b \ A.
\]
Against this, HOL's calculi are not decidable, and are -- by G\"{o}del's incompleteness theorem -- incomplete.

Two common ways in which the classical logics (in particular, FOL) are augmented are, first, by the addition of set theoretical axioms and, second, by the addition of modal operators.  The first allows the approximation of higher order logic while maintaining advantages of first order logic; the second allows logic to be applied to modalities, such as knowledge, belief, or time.

Set theoretical axioms allow the definition of new symbols and operations on both predicates (e.g.\ $\in$ and $\subseteq$) and functions (e.g.\ $\cup, \cap$ and $\emptyset$).\footnote{Constants such as $\emptyset$ are considered as a special case of functions, nullary functions -- functions that do not take any argument.}  They also allow the specification of properties of sets (e.g. $a \notin X$).  Adding set theoretical axioms to FOL allows it to weakly simulate HOL: functions can be expressed as relations over $X \times X$ that are left-total and right-unique; predicates are expressed as sets.  While HOL is still more expressive than FOL augmented by set theory (e.g., FOL cannot express inductive arguments), HOL's incompleteness means that there are true statements that can be expressed in HOL but which may not have finite proofs.  As FOL augmented by set theory uses FOL, it remains complete by using FOL's complete calculus.

Modal operators -- such as `next' and `until' -- allow the consideration of \textit{modes} (or \textit{states} in economic parlance).  \textit{Linear temporal logic} (LTL) is a popular simple modal logic, modelling states in a linear fashion, thus excluding the consideration of multiple possible future states. Kamp's theorem established the equivalence of LTL with a first-order logic.  Another first-order approach to modelling states is the \textit{situation calculus}~\citep{McCarthy-Hayes-69}, which allows expression of states and the temporal development of systems in first-order logic by representing the state as an extra argument of the formulae (e.g., that agent $i$ has $\pounds 10$ in state $s_0$ can be expressed as $\mathit{has}(i,10,s_0)$).  By referring to the state absolutely, rather than in relation to other states, the problem can be expressed in standard FOL without recourse to specialized modal relations.

Our final level of distinction is the domain of the problem; this level will allow us to present concrete examples of the preceding.  Table \ref{ta:ded-log} depicts these dimensions within deductive reasoning systems.

\begin{table}[h]
\begin{center}
  \begin{tabular}{c||c|c}
     & decidable & undecidable \\
    \hline\hline
    logic            & SAT, CSP; description logic & ITP, ATP \\
    \hline
    computer system  & model checking & program verification \\
  \end{tabular}
  \caption{Mechanized reasoning using deductive logics}
  \label{ta:ded-log}
\end{center}
\end{table}

\paragraph{Decidable logic:} In Table \ref{ta:ded-log}, the \textit{decidable logic} cell refers to decidable calculi as applied to logical problems.

\textit{\href{http://en.wikipedia.org/wiki/Boolean_satisfiability_problem}{Boolean satisfiability problems} (SAT)} are among the simplest canonical problems in propositional logic.  They specify a (finite) set of statements about a (finite) set of propositional variables, and ask whether there exists an assignment of values (i.e.\ true and false) to each of those variables that simultaneously satisfies all of the statements.

In SAT problems, clauses of Boolean variables are typically expressed in \textit{\href{http://en.wikipedia.org/wiki/Conjunctive_normal_form}{conjunctive normal form}}, conjunctions ($\wedge$) of disjunctions ($\vee$) such as
\begin{equation} \label{eq:log-equiv}
  \left( \neg p \vee q \right) \wedge \left( p \vee \neg q \right);
\end{equation}
where $p$ and $q$ are Boolean variables, evaluating either to true or false.\footnote{The sentence given here is logically equivalent to $p \equiv q$, an equivalence exploited by \citet{ta-li-11} in their search for uniqueness conditions in bimatrix games.}
Revisiting the example that in auctions the non-winning player pays nothing, equation (\ref{eq:FOL-bid}) can be translated for a finite number of bidders (here, three) to a propositional logic formula,
\def\wins#1{wins#1}\def\pay#1{payZero#1}
\begin{equation} \label{eq:sat-ex-prop}
 (\mathit{wins1}\lor \mathit{\neg pays1})\land (\mathit{wins2}\lor \mathit{\neg pays2}) \land (\mathit{wins3}\lor \mathit{\neg pays3});
\end{equation}
stating for each of the three players separately that they win or pay nothing.

Any formula in propositional logic can be expressed in this form, as can any formula in first-order logic when the domain is restricted to a concrete finite domain (such as three bidders in an auction). A SAT solver is used to try to assign the variables such that all of the clauses are true.
 For instance, assigning $\mathit{wins1}$ and $\mathit{pays1}$ to $\mathit{true}$ and the other predicates to $\mathit{false}$ shows that the single formula (\ref{eq:sat-ex-prop}) is satisfiable.

SAT problems are $\mathcal{NP}$-hard \citep{kar-72}, requiring -- in the worst case -- trial of every possible input.  Thus, while the logic and calculi involved are simple, SAT problems may not be computable in practice except in small cases.  However, techniques have been developed so that SAT solvers are able to solve typical cases very quickly. One application area of SAT solvers are model checkers, as described below.

\href{http://en.wikipedia.org/wiki/Constraint_satisfaction_problem}{\textit{Constraint satisfaction problems (CSP)}} are triples, $\langle V, D, C \rangle$, where $V$ is a set of variables, $D$ their domain, and $C$ the constraint set.  In CSPs, the variables may take on more values than in Boolean satisfiability's binary assignments.  For example, an $\mathit{hours}$ variable might take one of twelve values.  While apparently richer, CSPs can be reduced to SATs by suitable definition of additional auxiliary variables.\footnote{See \citet{bo-ha-zh-06} for a comparison of SAT and constraint programming.} 

The third example of decidable calculi applied to logical problems that we consider are \textit{description logics}.  These are central to automated reasoning about concept hierarchies in classification (or ontological) tasks.  One of their most important applications is to the \href{https://en.wikipedia.org/wiki/Semantic_Web}{\textit{semantic web}}, which allows computers to extract semantic information from web pages. As a simple example, semantically enabled web searches could recognize that $x^2 + y^2 = z^2$ and $a = \sqrt{c^2 - b^2}$ were both statements of Pythagoras' theorem.\footnote{See \citet{Lange:OntoLangMathSemWeb} for a more in-depth discussion of applications of semantic web technology to mathematics.}

\comment{Simple logics of this type
deal with simple combinations of concepts such as union or
intersection of concepts. More advanced logics allow for more involved
definitions such as a $\mathit{isFather}(q)$ means $\exists p . \mathit{father}(q,p)$. In the field it has been studied in great depth what can
be added in expressive power without sacrificing decidability. Also there is a good understand of the complexities for answering certain questions in different logics (for instance, assume that $\mathit{isMother}(q)$ and $\mathit{isParent}(q)$ are defined analogously, what is the time complexity for answering the question whether
$\mathit{isParent}(q)$ is equivalent to $\mathit{isFather}(q)\lor \mathit{isMother}(q)$).}


\paragraph{Model checking:}
Model checking~\citep{Clarke-et-al.86,Clarke-et-al.94} builds finite models
to describe computer hardware systems or simple software systems and then tests their properties. Typical questions include whether certain states of the system can be reached, or whether information is flowing properly through a circuit design.

Such models are typically expressed as \textit{finite automata}.  A finite automaton can model either a finite system or an infinite system if abstraction allows the infinite state space to be simplified to a finite one.\footnote{For example, in proofs involving real numbers, it may suffice to reduce an infinite number of possible values -- which cannot be handled by a decidable calculus -- to a trinary partition defined by $>$, $<$ and $=$.  See \citet{BCMDH-90} for an application to large, complex microprocessor circuits.}  Then the model is systematically checked for desired properties, e.g.\ by using SAT solvers. Viewing digital computer chips as a set of Boolean statements allows them to be modeled as \textit{decidable computer systems} allowing, in turn, SAT solvers to automatically verify their properties.  Since the mid-1990s, Intel has used formal methods to formally prove properties like `this chip implements the IEEE division standard' following an embarrassing and costly recall of a Pentium chip that was discovered not to properly implement IEEE floating point division \citep{harrison-sfm}.  No further such problems have been reported since then.\footnote{With chip design becoming more and more sophisticated, the reasoning in the  verification needed to become also more sophisticated.  Thus, HOL theorem provers such as HOL-Light are now also used for hardware verification.}


\paragraph{Undecidable logic: } The upper right cell in Table \ref{ta:ded-log} refers to the application of undecidable calculi to logical problems.  The two types of mechanized reasoning mentioned here, \textit{interactive theorem proving} (\textit{ITP}) and \textit{automated theorem proving} (\textit{ATP}) have traditionally been equated with theorem proving, but seen as distinct, with the former involving more steering from a human user than the latter.  Stereotypically, an ITP system could check an existing proof, while an ATP system could suggest steps in a proof or, in some cases, a whole proof.  In practice, the distinction between the two has decreased, with ITP systems implementing ATP procedures.\footnote{\citet{harrison-ab} noted that ITP may be preferred to ATP, as -- in working more closely alongside human reasoning -- it may be better at developing human understanding.}

The traditional identification of theorem proving with work in these areas owes partly to some high profile successes in pure mathematics, the focus of the most hope in mechanized reasoning's early days.  The earliest major success was -- as might be expected in an emerging field -- not even a clear example of mechanized reasoning: in the 1970s, computers were used to carry out the exhaustive computations required to prove the four-color map theorem \citep[q.v.][]{ap-ha-77,ap-ha-ko-77}.  Here, the computers were used to perform simple (algebraic) calculations, rather than to (logically) `reason'.  More recently, mechanized proof checkers have confirmed these results formally \citep[q.v.][]{gon-08}.\footnote{\citeauthor{gon-08}'s team has now also formally checked the Feit-Thompson Odd Order Theorem \citep{gonthier:hal-00816699}.}

The first major mathematical result to be established by mechanized reasoning -- rather than `mere' calculation -- was Robbins' conjecture that two bases for Boolean algebras are equivalent.  While appearing to be a beguilingly simple problem, it remained unresolved for 60 years, becoming a favourite of Tarski, who set it as an open problem \citep[q.v.][p.~245]{he-mo-ta-71}.  One of the complicating factors of the conjecture was that the only known example of a Robbins algebra was also a Boolean algebra, reducing the evidence base that mathematicians could use to form intuitions about the problem.  Nonetheless, in the late 1990s, \citet{mcc-97} was able to pose the problem in a way that allowed EQP, an automated theorem prover related to his well-known Otter prover, to generate -- not just check -- a 17-step proof, later reduced to eight steps \citep{mcc-97}.\footnote{\citet{dah-98} manually reworked EQP's proof to provide a more human-readable proof.}

Perhaps the highest profile success of mechanized reasoning in pure mathematics is the solution to Kepler's conjecture that there is no denser packing of spheres in $\mathbb{R}^3$ than the face-centred cubic.  \citeauthor{hal-05}' original proof was 120 pages long (excluding computer code that exceeded 500MB), requiring a team of 12 referees five years to become ``99\% certain'' that it was correct.  Unsatisfied with this standard, \citeauthor{hal-05} founded \textit{Project Flyspeck} to establish a fully formal proof of the conjecture \citep{Hales:DenseSpherePackings12}.  In August 2014, the project was completed \citep{hales2015formal}, close to \citeauthor{hal-05}'s original estimate of 20 person-years \citep{av-ha-14}.

More mundanely, ITP has been used to translate existing human proofs into formal proofs that are sufficiently detailed that a computer can mechanically verify them: as of January 2016, 91 of the \href{http://www.cs.ru.nl/~freek/100/}{`top 100' mathematical theorems} on a list maintained by \citet{wie-xx} had been formalized.\footnote{Exceptions include Fermat's last theorem.}  While most of these are considerably less spectacular than the examples cited above -- in which theorem provers have been used to help convince mathematicians as to the validity of major, new results -- the gradual accretion of small proof libraries builds a foundation for applying ATPs more widely.

The distinction between high-profile, major theorems and lower-profile bodies of theory has been suggested as a reason that ATP has yet to fulfil its early hopes: \citet{Buchberger06} noted that human mathematicians typically do not try to prove isolated theorems but explore a whole theory, thereby building up valuable intuition which helps them in proving related theorems.  Additionally, \citet{Newell81} stated that standard theorem proving techniques -- while often highly efficient -- do not make use of advanced human approaches (as described in \citeauthor{Polya45}'s books) such as simplifying a problem to one they can solve; applying the simplified solution to the original problem may still be very hard, but the intuition gained by solving the simplified problem may help solve the original problem.\footnote{Conversely, \citet{dic-11} observed that the `resolution' inference rule \citep{rob-65}, central to mechanized reasoning, ``was not based on any known human practice and was in fact difficult and counterintuitive for humans to understand''.  Indeed, reviewing mechanized reasoning since resolution, \citeauthor{rob-65} lamented that it may have harmed mechanized reasoning by contributing to a parting of ways between human mathematicians and mechanized reasoners \citep{dic-15}.}

\paragraph{Program verification} Table \ref{ta:ded-log}'s lower right cell corresponds to software engineering's \textit{program verification}, reasoning about software systems.  This can be highly complex in the case of complex programs.  Within program verification, traditional proof approaches have sought to prove that the software correctly implements properties specified in the design brief.  As such proofs are very costly, full correctness proofs that seek to verify all desired properties of the code, are done only for `mission critical' systems~\citep{Silva-et-al.08}.

Some well known examples of program verification have come from transport and finance: in code controlling automated commuter rail systems, theorems that no two trains occupy the same location at the same time have been proved; within financial transactions software, theorems that transactions do not create or destroy value, but merely transfer it, have also been proved~\citep{wo-la-bi-fi-09}.  More recently, a compiler for the C programming language has been formally verified \citep{boldo:hal-00743090}.  These techniques are becoming more mainstream: in 2013, Facebook acquired Monoidics, a start-up firm applying theorem proving to software code analysis; in 2015, another start-up, Aesthetic Integration beat 600 competitors to win first prize in UBS' Future of Finance Challenge for its ability to automatically prove failure or compliance in financial algorithms.\footnote{Their entry formally defined a UBS `dark pool' and a set of SEC regulations which the SEC had found the dark pool in breach of.  Aesthetic Integration was able not only to verify the dark pool failure found by the SEC, but discovered that its order prioritization failed to satisfy transitivity \citep{ih-pa-15}.}

Historically, program verification has been conducted as a \textit{post mortem}: given existing code, program verification determines whether or not it is correct.  More recently, \textit{code extraction} techniques have been developed to generate code that provably implements the desired properties.

\subsection{Inductive reasoning}

As noted above, both inductive and deductive reasoning date back at least to Aristotle, but the former is not sound, while the latter has been the focus of the mechanized reasoning community.  The distinction between the two -- as well as the utility of each -- was expressed by \citet[p.~vi]{Polya54}, who referred to deductive reasoning as \textit{demonstrative reasoning}, and inductive reasoning as \textit{plausible reasoning}:
\begin{quote} 
  We secure our mathematical knowledge by {\em demonstrative
    reasoning}, but we support our conjectures by {\em plausible
    reasoning\/} \dots Demonstrative reasoning is safe, beyond
  controversy, and final. Plausible reasoning is hazardous,
  controversial, and provisional. \dots\\
  In strict reasoning the
  principal thing is to distinguish a proof from a guess, a valid
  demonstration from an invalid attempt. In plausible reasoning the
  principal thing is to distinguish a guess from a guess, a more
  reasonable guess form a less reasonable guess. $\ldots$ [plausible
  reasoning] is the kind of reasoning on which [a mathematician's]
  creative work will depend. 
\end{quote}

Inductive systems seek to derive general statements based on a finite number of statements (e.g.\ if $A_1$ is true, and $A_2$ is true, and so on up to $A_N$ for some finite $N$, then $A_n$ is true for all natural numbers $n$).\footnote{Inductive reasoning is distinct from mathematical induction, which involves proving $A_0$ and that $A_{n+1}$ is true given $A_n$. Mathematical induction is a sound \textit{deductive} method.}  This sort of reasoning is immediately familiar to us when we reflect on how we form conjectures: we expect the sun to rise tomorrow without any understanding of astrophysics; this expectation, though, may lead to the formation of conjectures about astrophysics.  However compelling the weight of evidence, inductive reasoning is not sound -- as may be demonstrated by single counterexamples.  In number theory, \href{https://en.wikipedia.org/wiki/Euler\%27s_sum_of_powers_conjecture}{Euler's attempted generalization of Fermat's last theorem} remained open for two centuries until a computer found a counterexample.\footnote{Euler's conjecture states: let $n$ and $k$ be integers greater than one, and let $a_1, \ldots a_n$ and $b$ be non-zero integers; then $\left( \sum_{i=1}^n a_i^k = b^k \right) \Rightarrow \left( n \ge k \right)$.  The first known counterexample, found by computer, is $27^5 + 84^5 + 110^5 + 133^5 = 144^5$ \citep{la-pa-66}.}  In game theory, \citeauthor{vN-M} conjectured that stable sets (`solutions' in their parlance) always existed; it took almost a quarter-century for counterexamples to be found \citep{luc-68}.

Inductive reasoning may be used for \textit{theorem discovery}, whereby regularities in observed data are used to form conjectures to test.\footnote{One of the most dynamic subfields of AI currently is \textit{machine learning}.  Some definitions are agnostic as to how the machines learn -- e.g.\ whether deductively or inductively -- while, perhaps more typically, others link machine learning more closely to inductive reasoning.  Some of the highest profile applications of machine learning are statistical, positing rules that fit the existing data well, rather than perfectly.}  \comment{Famously, \citet{Polya54} described how mathematical concepts and relationships can be found by extrapolating beyond the actual definitions and restrictions, that is, by performing an initially unsound generalization. He gave as an example Euler's generalization from polynomials to infinite power series and the analogy of representing the zeros of a polynomial in a product $a_n(x-\alpha_1)(x-\alpha_2)\cdots(x-\alpha_n)$ on the one hand and those of an infinite power series in a product $(1-{x^2\over \pi^2})(1-{x^2\over {4\pi^2}})(1-{x^2\over {9\pi^2}})\cdots$ on the other hand. The latter has exactly the roots $\pi, -\pi, 2\pi, -2\pi, 3\pi, -3\pi, \ldots$ hence corresponds to $\tfrac{\sin x}{x}$ \citep[p.18]{Polya54}. Although not warranted by the theory developed so far, the generalization from finite polynomials to infinite series generates a good hypothesis.} \todo{CR@MK: I've cut the Polya example as -- while I like it -- I still don't understand it, and don't think it central to an intro to mech reasoning.\\ MK@CR: I am not sure, do you want to replace it, or leave it out? I re-read the paragraph in Polya's book, I don't think that there is much more to it.}

Mechanized inductive reasoning dates back to two systems built in the 1970s and 1980s to discover new conjectures, AM (Automated Mathematician)~\citep{Lenat76} and Eurisko~\citep{Lenat83}. These were able to detect conjectures such as the unique prime factorization theorem and Goldbach's conjecture.\footnote{The prime factorization theorem states that any positive integer has a unique decomposition as the product of primes.  Goldbach's conjecture states that every even integer beyond two can be expressed as the sum of two primes.} \comment{The AM system has been criticized because of the close link between the system and its implementation language Lisp. For instance, the concept of a natural number was discovered by AM through the discovery that some lists share a common property, namely their lengths. This concept was then called a natural number. This step was criticized as not genuine since the step would be easy since the concept of a list was already built in the system, and hence the statement that AM discovered natural numbers would not be warranted.} \todo{CR@MK: I've cut the AM/Lisp criticism as it doesn't make sense to me, and doesn't seem important to understanding what follows.} \comment{In contrast Eurisko used explicit knowledge about the domain and about heuristics in form of rules that speak about the relevance of concepts and potential relationships between them.}
\todo{CR@MK: ditto.}
The systems use certain measures of interestingness for concepts. For instance, concepts that are always true or always false are not interesting. However, if a concept is true for a significant proportion of examples (such as divisibility by only 1 and the number itself) then this is considered as an interesting concept (`primality' for divisivility by only 1 and the number itself).\footnote{\citeauthor{dic-11}'s case study of the Argonne National Laboratory's AURA system noted that, while ``the capacity to identify what was `promising' or `interesting' was precisely one of those unautomatable human abilities \ldots the Argonne practitioners decided what was important on the basis of extensive experimenting with AURA.''}

Lenat's work was continued by Colton in the HR (Hardy-Ramanujan) system~\citep{colton99}, where more advanced measures for interestingness were developed. For instance,
\begin{quote}
  The novelty measure of a concept calculates how many times the categorisation produced by the concept has been seen. For example, square numbers categorise integers into two sets: $\{1,4, 9, \ldots\}$ and $\{2, 3,5,\ldots\}$. If this categorisation had been seen often, square numbers would score poorly for novelty, and vice-versa. \citep{colton00}.
\end{quote}
Another important advance in Colton's work is that the HR system weeds out simple conjectures, namely those that can be easily verified or falsified by automated theorem provers.\footnote{See also the introduction of \citet{ta-li-11} for a brief review of the history of mechanized theorem discovery; a lengthier review is available in \citet{tan-10}.} One of the successes of HR was that it invented the concept of `integers with a square number of divisors' which was added to Sloane's Encyclopedia of Integer Sequences.\footnote{\url{https://oeis.org/}}

\section{Mechanized reasoning for economic problems} \label{se:amre}

Over the past decade, computer scientists have become interested in economic problems -- often publishing economically novel and interesting results, but almost entirely within the computer science literature.  This section reviews that literature, focusing on the applications to social choice and auction theory.  We structure this survey primarily according to the problem domain within economics, and only secondarily according to our classificatory scheme, in order to focus on the insights into economic problems made possible by these techniques, rather than the techniques themselves.

Table \ref{ta:app-ded-log} places the papers reviewed in this section into our original classificatory scheme.  This classification is imperfect. For example, \citet{ta-li-09} and \citet{ge-en-11} both used propositional logic solvers (and, therefore, deductive reasoning), but used them to discover new results -- which we have associated, above, with inductive reasoning.  Papers like this therefore span historical distinctions.

\comment{However, there is an important difference, the discovery of a relationship such as Goldbach's conjecture in a system such as Eurisko or HR does not say anything about the correctness of the conjecture (and indeed Goldbach's conjecture is at the time of writing this article still an open problem). By contrast, in \citet{ta-li-09} and \citet{ge-en-11}, the correctness of the newly established theorems is guaranteed: the computer verifies the computationally challenging base case and the hand-written inductive argument generalizes the relationship for arbitrary other values.}

Social choice has been mechanized reasoning's main point of contact with economics, making it a convenient lens for illustrating mechanized reasoning.  Auction theory is, we feel, promising as a new point of contact between mechanized reasoning and economics, due both to the technical parallels between social choice (where mechanized reasoning has proved fruitful) and mechanism design (q.v.\ \citet{ren-01}), and to auctions' importance as allocation mechanisms.

\begin{table}[h]
\begin{center}
  \begin{tabular}{p{0.12\linewidth}||p{0.39\linewidth}|p{0.39\linewidth}}
     & decidable & undecidable \\
    \hline\hline
    \multirow{3}{*}{logic} & \citet{ge-en-11}, \citet{br-ge-15}: SAT & \citet{nip-09}, \citet{wie-07}, \citet{wie-09}, \citet{LangeEtAl:CompProvAuctThy13}: ITP \\
                           & \citet{ta-li-09}: SAT, CSP & \citet{gr-en-jpl}: ATP \\
                           & \citet{BaiTadjouddineGuo2014}: description logic & \\
    \hline
    computer system  & \citet{xu-ch-07}, \citet{ar-es-no-ra-si-05}, \citet{ta-gu-va-09} 
    : model checking & \citet{CKLR:SASI-EC15}: code extraction \\
  \end{tabular}
  \caption{Some applications of mechanized reasoning to economic problems}
  \label{ta:app-ded-log}
\end{center}
\end{table}

\subsection{Social choice}\label{sec:socialChoice}

\citeauthor{gea-05}' three brief and distinct proofs of Arrow's impossibility theorem -- that, for three or more alternatives and a finite set of agents, there is no social choice rule satisfying unanimity (\textit{UA}), independence of irrelevant alternatives (\textit{IIA}) and non-dictatorship (\textit{ND}) -- served as the mechanized reasoning community's entr\'{e}e to economic problems: social choice was novel to this community, yet used familiar structures -- particularly linear orders -- and the three proofs by \citet{gea-05} gave the mechanized reasoning community an opportunity to attempt to compare the relative difficulty of encoding those proofs for computers.

One primitive measure of the relative difficulty of formal proofs is to compare their size to that of human proofs.\footnote{The easiest way of determining the size of a formal proof is by counting lines of source code.  In Section~\ref{se:mpv} we discuss a less biased measure, the de Bruijn factor.}  Table \ref{ta:gea-proofs} reports on the relative sizes of \citeauthor{nip-09}'s proofs in Isabelle -- a higher-order logic theorem prover -- and \citeauthor{wie-09}'s proof\footnote{\citeauthor{wie-09} justified his decision to formalize only \citeauthor{gea-05}' first proof by noting that they became successively more abstract, making the first the most challenging as, generally ``abstract mathematics is easier to formalize than concrete mathematics'' \citep{wie-09}.} in Mizar -- a set theoretic proof checker, which augments first-order logic by the axioms of Tarski--Grothendieck set theory.\footnote{The advantage of Tarski--Grothendieck set theory over Zermelo-Fraenkel is that the former only requires finitely many axioms to axiomatize sets.} \citet{nip-09} attributed the greater length of the Mizar proofs to Isabelle's ``higher level of automation'' -- something to which we return in our Isabelle proof of Vickrey's theorem.

\begin{table}[h]
  \centering
  \begin{tabular}{lll}
    & $1^{\text{st}}$ proof & $3^{\text{rd}}$ proof \\
    \midrule
    Paper \citep{gea-05} & 1 page & 1 page \\
    Isabelle \citep{nip-09} & 350 lines (6 pages) & 300 lines \\
    Mizar \citep{wie-07,wie-09} & 1100 lines & \\
  \end{tabular}
  \caption{Relative lengths of human and machine proofs of Arrow's theorem}
  \label{ta:gea-proofs}
\end{table}

\citeauthor{nip-09}'s formalization attempts began with \citet{gea-01}, a working paper that preceded the published version \citep{gea-05}.  In seeking to formalize the first proof, he discovered a statement in one of the lemmas that required a 20 line auxiliary proof to properly establish.  Further, a relationship between a pivotal voter and a dictator only ``hinted at'' in the original text required elaboration.  \citeauthor{nip-09} did not discover any errors in this first proof.  Similarly, \citet{wie-09} reported on missing cases, but no ``real errors''.

As to the third proof, \citeauthor{nip-09} found two instances of omitted material in its central lemma, preventing him from formalizing the proof.  \citeauthor{nip-09} presented these concerns to \citeauthor{gea-01} by e-mail; both concerns were resolved in \citet{gea-05}.\footnote{Mechanized reasoning can identify omissions by forcing close scrutiny.  This, of course, is also possible without mechanical support.  For example, in the matching literature, \citet{ay-so-13} identified a hidden assumption in \citet{ha-mi-05} -- which they view as ``widely considered to be one of the most important advances of the last two decades in matching theory'' -- without which many of their results fail to hold.  The oversight arose from ``an ambiguity in setting the primitives of the model''.  This ambiguity would likely have been detected by a mechanized reasoner as well.}

Both \citeauthor{nip-09} and \citeauthor{wie-09}'s proofs were written by the authors themselves, and are therefore examples of ITP.  By contrast, \citet{gr-en-jpl} sought to, first, restate Arrow's theory in FOL and, then, to automatically generate a proof for it.\footnote{\citet{gr-en-jpl} is also a good guide to related work on formalizing results in social choice.}  Expressing Arrow's theory in FOL presented  the challenge that quantifying over all possible linear orders of agents' preference profiles appears to be a second-order quantification as it involves quantifying over agents, alternatives, and the agents' preference profiles. \citeauthor{gr-en-jpl} addressed this by adopting the approach taken in \citet{ta-li-09}, namely to apply the situation calculus (mentioned in section~\ref{subsection:deductive-reasoning}) for the representation. Thus, they could present a first-order formalization of the requisite axioms, $T_{ARROW}$, allowing them to restate Arrow's theorem as:
\begin{theorem}[Arrow \`{a} la \citet{gr-en-jpl}]
  $T_{ARROW}$ has no finite models.
\end{theorem}
A model in this sense is an instantiation (or example) of the variables used in the theory.  For Arrow's theorem, the variables include $N$ (the set of agents), $A$ (the set of alternatives), the set of the agents' preference profiles, and the set of social welfare functions (SWFs) mapping from such profiles to a social preference.  In the two-agent, three-alternative case, that $T_{ARROW}$ ``has no finite models'' means that none of the $6^{36}$ possible SWFs satisfy the theory's axioms.\footnote{There are a total of $36$ preference profiles in the domain, and six orders in the range, yielding a total of $\prod_{i=1}^{36} 6$.}  The theorem claims this property for any finite number of agents, and any finite number of alternatives in excess of three.

FOL's completeness allows any property of the system to be explicitly derived.  However, the second problem with FOL encountered by \citeauthor{gr-en-jpl} is that FOL is unable to express finitude, for the same reason that it cannot express induction: intuitively, HOL defines finitude by considering the complement of the infinite, which it can define by induction on the natural numbers.  Thus, formulating Arrow's Theorem in FOL requires a separate formulation for each $\left| N \right|$.  Similarly, proofs of Arrow's theorem in FOL may differ for each $\left| N \right|$.  Thus, \citeauthor{gr-en-jpl}' attempts to use a first-order theorem prover to automatically generate proofs of Arrow's theorem failed outside of minimal cases.\footnote{They used Prover9, a successor to Otter, and -- therefore -- a close relative of the system that found the proof of Robbins' conjecture \citep{mcc-97}.}

Independently of \citeauthor{gea-05}' proofs, \citet{suz-00} had presented an induction proof of Arrow's impossibility theorem for a base case of two agents and $\left| A \right|$ alternatives; an induction result then demonstrated its truth in general.  This motivated \citet{ta-li-09} to manually derive a second induction result in the number of agents.  Proving the impossibility in a two-agent, three-alternative base case, would -- by their two induction lemmas -- cause it to hold in general.  They computationally exhausted this base case in two different ways.

First, they expressed the problem as a Boolean SAT problem.  \citeauthor{ta-li-09} then used the situation calculus, which allows many of the problem's symmetries to be efficiently dealt with by the action of swapping arguments, to reduce the number of variables needed in the base case to $35{,}973$ in $106{,}354$ clauses. These are too many cases to check manually.  However, using the SAT solver Chaff2 they could show the inconsistency between the three basic axioms in less than a second on a desktop computer.

Second, \citeauthor{ta-li-09} expressed the problem as a CSP, in which $V$, the set of variables, consists -- in their base case -- of $36$ preference profiles; $D$, their domain, of six linear orderings for each profile; and $C$, their constraint set, of the \textit{UN} and \textit{IIA} axioms.  As the base case implies $6^{36} \approx 10^{28}$ possible SWFs -- far too many to be feasibly generated -- the authors used the (first-order) logical programming language \href{http://en.wikipedia.org/wiki/SWI-Prolog}{Prolog} to generate all SWF satisfying the constraints of \textit{UN} and \textit{IIA}.\comment{\footnote{Prolog uses a restricted first-order language which allows for its efficient execution. Expressions such as $p(x)\leftarrow q(x)\land r(x)$ can be used in order to establish queries such as $p(a)$ by instantiating the variable $x$ by the constant $a$ and then trying to establish $q(a)$ and $r(a)$ separately.  A fixed execution order makes the logic expressions then computational. Expressions that introduce choice in the computation such as $p(x)\leftarrow q(x)\lor r(x)$ are not allowed in Prolog.}} Running in less than a second on a desktop computer, their Prolog code generated two SWFs, both of which were also dictatorial.

A similar approach yielded the Muller-Satterthwaite theorem, and Sen's Paretian liberal result, among others.\footnote{See \citet{gei-10} for a more complete list.}

When implementing the CSP, the authors noticed that imposing even just the \textit{IIA} constraint reduced the set of SWFs from $6^{36}$ to 94.  By inspecting these manually, \citet{ta-li-09} posited a new theorem that implies both Arrow's and Wilson's.  Before stating it, note that a social order is \textit{inversely dictatorial} if it ranks elements in the opposite way to at least one agent; the \textit{Kendall tau} distance between two orderings is the number of pairs on which they disagree.  Then:

\begin{theorem}[\citet{ta-li-09}]
  If a social welfare function $W$ on $\left(N, A \right)$ satisfies \textit{IIA}, then for every subset $Y$ of $A$ such that $\left| Y \right| = 3$,
  \begin{enumerate}
    \item $W_Y$ is dictatorial, or
    \item $W_Y$ is inversely dictatorial, or
    \item \label{it:Ktd} The range of $W_Y$ has at most 2 elements, whose [Kendall tau] distance is at most 1.
  \end{enumerate}
\end{theorem}

As an example of an SWF accepted under condition \ref{it:Ktd} of theorem, consider the function that always prefers the first alternative to the second, always prefers the first to the third, and prefers the second to the third alternative unless both agents prefer the third to the second. This is neither dictatorial nor inversely dictatorial: the agents' preferences for the first item are ignored; there are only two elements in its range (e.g.\ $a \succ b \succ c$ and $a \succ c \succ b$), the distance between which is one.\footnote{Represent preferences over three objects as a three-digit binary character, the first indicating whether $a \succ b$, the second whether $a \succ c$ and the third whether $b \succ c$.  There are six permissible three digit numbers, $000, 001, 011, 100, 110$ and $111$, after eliminating the two cyclical ones.  \textit{IIA} then requires that each digit in the social preference is a function of the corresponding digits in the individual preferences alone.  The 1-distance condition then allows only one of those digits to vary.}  As \citeauthor{ta-li-09} noted, the third case of their result violates Arrow's original non-imposition axiom, which requires that the SWF be surjective, mapping to every possible value in its range.

Of the 94 SWFs satisfying \textit{IIA}, there are 84 of the sort described above, 6 constant SWFs (one for each ordering), two dictatorial functions, and two inversely dictatorial functions.

As before, the theorem is established by exhaustive computation on the two-agent, three-alternative base case, and then extended to arbitrary finite domains by the manually-derived induction lemmas.  \citet{ch-se-14} observed that, as far as they were aware, this is the ``only Arrow-type result in the literature that does not use an axiom other than \textit{IIA}'', an achievement that they believe ``could not have been conjectured without computational aid''.\footnote{In private correspondence, Sen has conjectured that the result of \citet{ma-zh-94} linking Wilson's and Arrow's theorems may be an immediate consequence of Tang and Lin's.}

Social choice is replete with characterization and impossibility results.  \citet{ge-en-11} applied the \citet{ta-li-09} approach to the problem of ranking sets of objects \citep{ka-pe-84}, for which \citet{ba-bo-pa-04} supplied almost 50 possibly desirable axioms.\footnote{\citet{gei-10} had initially attempted an approach more akin to \citet{gr-en-jpl}, seeking to derive an automated proof of the \citeauthor{ka-pe-84} theorem using three different first-order theorem provers; none of them was able to derive a proof after 120 hours of CPU time on 2.26 GHz machines with 24 GB RAM.}

Rather than deriving an induction lemma for every base case of interest, they derived a broadly applicable induction theorem based on model theory's \href{https://en.wikipedia.org/wiki/\%C5\%81o\%C5\%9B\%E2\%80\%93Tarski_preservation_theorem}{{\L}o\'{s}--Tarski preservation theorem}, which describes when properties ($\varphi$, below) are retained in substructures, namely essentially when the theory can be expressed using universal quantifiers in the form $\forall x\,.\,\varphi$.\footnote{As a trivial example, the property that a structure contains three distinct elements cannot be preserved in substructures with fewer than three elements.}

Furthermore, as they wished to distinguish between individual alternatives, sets of preferences, and preference orders the authors used a many-sorted FOL.  Many-sorted FOL also allows relations (including set inclusion or union) to be defined on one domain that do not hold on the other.

\citeauthor{ge-en-11} then encoded 20 axioms drawn from \citet{ba-bo-pa-04} in their many-sorted FOL.  As their induction result translated impossibilities generated on small, finite domains to full-blown impossibility results, they took advantage of these concrete, finite base cases to re-write the axioms in propositional logic (using the kind of rewriting that transformed formula (\ref{eq:FOL-bid}) to formula (\ref{eq:sat-ex-prop}) in section~\ref{subsection:deductive-reasoning}). This, in turn, allowed them to use SAT solvers to search for subsets of axioms which generate impossibility results in these base cases; once found, the induction theorem generalized them to full impossibility results.  Doing so for all base cases up to sets of eight items yielded 84 impossibility theorems from about one million combinations.\footnote{Resource constraints limited them to eight items and 20 axioms.  They derived their results in about one day.}

Their results included known results (e.g.\ those of \citet{ka-pe-84} and \citet{ba-pa-84}); variations on known results, typically formed by strengthening axioms to reduce the impossibility's minimal domain; direct consequences of other results (as they did not prune implications of existing impossibilities); a trivial contradiction between the axioms of uncertainty aversion and uncertainty appeal; and -- perhaps most interestingly -- new theorems.  These last resolved an open question in the literature, which we now describe.

Letting $\succ$ (resp.\ $\succsim$) denote strict (resp.\ weak) preference on individual choice objects (denoted by lower case letters), and $\triangleright$ (resp.\ $\trianglerighteq$) strict (resp.\ weak) preference on sets of objects (denoted by capital letters), \citet{bo-pa-xu-00} presented a theorem characterizing the min-max ordering in terms of four axioms.  The min-max ordering is defined as
\[
  A \trianglerighteq_{\mathit{mnx}} B \Leftrightarrow \left[ \min \left\{ A \right\} \succ \min \left\{ B \right\} \vee
  \left( \min \left\{ A \right\} = \min \left\{ B \right\} \wedge \max \left\{ A \right\} \succsim \max \left\{ B \right\} \right) \right];
\]
where $\min \left\{ A \right\}$ is the minimal element of $A$ with respect to $\succsim$ and $\max \left\{ A \right\}$ the maximal element.  Thus, a set $A$ is weakly preferred under the min-max ordering to set $B$ iff either the worst element of $A$ is strictly preferred to that or $B$, or (when the worst elements are equally preferred) the best element of $A$ is weakly preferred to that of $B$.

The four axioms were:
\begin{enumerate}
  \item \textit{simple dominance},
  \[
    x \succ y \Rightarrow \left( \left\{ x \right\} \triangleright \left\{ x, y \right\} \wedge \left\{ x, y \right\} \triangleright \left\{ y \right\} \right)
  \]
  for all $x$ and $y$, so that a set consisting of a strictly preferred object is preferred to a set containing it as well as a strictly less preferred object, which -- in turn -- is preferred to a set consisting only of that less preferred object.

  \item \textit{independence},
  \[
    A \triangleright B \Rightarrow A \cup \left\{ x \right\} \trianglerighteq B \cup \left\{ x \right\}
  \]
  for all $A$ and $B$ and $x$ not contained in $A$ or $B$.  Thus, adding a single object to two sets ranked by strict preference does not reverse that ranking (but it may weaken it).

  \item \textit{uncertainty aversion},
  \[
    \left( x \succ y \succ z \right) \Rightarrow \left\{ y \right\} \triangleright \left\{ x, z \right\}
  \]
  for all $x, y$ and $z$, so that a set consisting only of an intermediately preferred object is strictly preferred to a set consisting of a strictly more favourable and a strictly less favourable object.

  \item \textit{simple top monotonicity},
  \[
    x \succ y \Rightarrow \left\{ x, z \right\} \triangleright \left\{ y, z \right\}
  \]
  for all $x, y$ and $z$ such that $x \succ z$ and $y \succ z$, so that -- if an object is strictly preferred to another -- a set containing it and a third object is strictly preferred to a set containing the less preferred object and the third object.
\end{enumerate}

\citet{arl-03} showed that the min-max ordering was, in fact, inconsistent with the independence axiom, and presented an alternative axiomatic basis for it.  \citet{ge-en-11} presented a complementary result to \citeauthor{arl-03}'s, finding a contradiction between the four original axioms at even four choice objects, thus establishing that the original four axioms are inconsistent, so cannot form the basis of any transitive binary relationship.

\citet{ge-en-11} also presented the first impossibility result in this literature not to use any dominance axiom.

In cases of interest, the authors were able to quickly derive manual proofs for the computationally discovered results.\footnote{For the min-max ordering inconsistency, the manual proof is about a half-page long.}

Finally, the large set of impossibility results allowed the authors to statistically consider the role of the various axioms.  For example, the linear order axiom appeared in all theorems; the `even-numbered extension of equivalence' and reflexivity occurred in none; `intermediate independence' occurred in all results for seven or eight choice items, but never for fewer than five choice items.

\citet{br-ge-15} extended the methodology of \citet{ge-en-11} by performing an initial encoding in HOL, and then deriving implications capable of expression in propositional logic for small base cases.  This allowed expression of more properties than was possible in the many-sorted FOL of \citet{ge-en-11}.  Thus, \citet{br-ge-15} could encode a neutrality axiom that \citet{ge-en-11} could not, but at the cost of generating exponentially many new variables, restricting the size of cases that could be computed.

\subsection{Auctions}

Applications of mechanized reasoning to auction design and implementation are less sophisticated than those to social choice.  Nevertheless, given auctions' practical importance, we expect that these will ultimately become more widespread.  This section surveys work in two separate areas -- applying mechanized reasoning to checking results in auction theory, and checking implementations of auction designs.

On the former, Vickrey's theorem has provided a basic testbed result.  Section \ref{se:mpv} illustrates in detail our Isabelle implementation.  It therefore complements \citet{LangeEtAl:CompProvAuctThy13}, which compared implementations of Vickrey's theorem in four different mechanized reasoners.

Conceptually, as higher-order logic is sufficient to express all concepts in auction theory, it is not challenging to represent basic results in auction theory using a higher-order logic theorem prover like Isabelle.  Doing so in more basic logics is both more conceptually challenging, and may offer more promise of automation.

In simpler logics, model checking can automatically establish properties of systems by exhaustively inspecting the system's state space.  \citet{ta-gu-va-09} used SPIN, a widely-used commercial model checker based on a linear temporal logic (LTL), to verify Vickrey auctions' strategy-proof\-ness property that bidders cannot do better than to bid their valuations.\footnote{\citet{ta-gu-va-09} did not seem to use the modal capabilities of SPIN; instead, the authors seemed to adopt SPIN as they wished -- in future work -- to be able to accept C code as input, and to reason about it; reasoning about computer programs in which variables can be set does require modal capability.}  They implemented two techniques to reduce the search space while verifying strategy-proofness for arbitrary bid ranges and numbers of agents: \textit{program slicing} removed variables irrelevant to the property; \textit{abstraction} discretized the domain of bids into a three-element domain, depending on whether a bid exceeded, equalled, or was less than an agent's valuation.  A manual proof was required to establish the abstraction's soundness.  Together, the two simplifications allowed strategy-proofness to be verified for any number of agents in a Vickrey auction in a quarter of a second.

The second branch of applications of mechanized reasoning to auctions has sought to establish properties of auction designs as implemented.  This is of interest for at least two reasons: first, even if theoretical properties of an auction are known, errors may be introduced when translating the auction from a design to an operational auction.  Second, and more commonly for modern auctions, practice may simply outstrip theory.  In both cases, mechanized reasoning can be used to reduce the likelihood that an auction will fail when run.

\citet{CKLR:SASI-EC15} used Isabelle to prove that a combinatorial Vickrey auction is soundly specified, in the sense of guaranteeing that -- whatever the bids received as input -- the output allocated only the available goods, at non-negative prices, and assigned a unique output to each input.  Furthermore, it implemented two parallel specifications of the auction, the first close to its standard paper specification, and the second a constructive one. Constructive definitions are essentially algorithmic descriptions.  By contrast, definitions in classical logics need only state properties of the defined object.  For instance, a classical definition of the maximum of a (non-empty) list of bids identifies an element of the list that is greater than or equal to every other element in the list.  A constructive definition would begin by noting that -- for a one-element list -- the maximum is the single element of the list; it would then proceed recursively by computing the maximum of the remainder of the list.  It would then return the larger of the two: the initial element, or the maximum of the remaining elements.

Isabelle was used to formally prove the equivalence of the two specifications.  While the constructive specification is less intuitive, its algorithmic nature allows Isabelle to automatically generate verified executable code from it.

Model checking has also been used to examine auctions for evidence of shill bidding.  \citet{xu-ch-07} used SPIN to define predicates corresponding to suspicious behaviour, including pushing prices to a reserve price before dropping out, and bidding on the higher priced of two identical goods. The model checker was then used to see whether the predicates were present in a finite dataset of actual bidding behaviour.

\citet{ar-es-no-ra-si-05} developed a toolkit to verify properties of multi-agent environments, with a traditional open outcry auction as their leading example.  Their toolkit implemented liveness checks to ensure that agents are not blocked (i.e.\ can bid in every round),
that each bidding round can be reached, and that the final bidding round is reachable from any other, as well as correctness of the bidding language (that is, that by following the rules, the system always remains in a defined state).
Their toolkit also includes a simulation tool that conducts a `what-if' analysis by performing a complete check of all cases.  While the authors themselves do not refer to what they do as model checking, that is what it most closely resembles.

Finally, \citet{BaiTadjouddineGuo2014} consider the question of how potential users of online auctions can trust the auctions' protocols.  They develop a protocol for specifying auction designs that can be read by Coq, a mechanized reasoner.  Future work building on this should eventually allow Coq to verify properties claimed for the auction.

\section{Blueprint of a formal proof of Vickrey's theorem} \label{se:mpv}

The preceding has provided an overview of mechanized reasoning, both in general, and as applied to economic problems.  This section provides a detailed description of how a mechanized reasoner is used in practice, in this case to verify a formal proof of Vickrey's theorem.  We use Vickrey's familiar theorem to focus attention on the formal proof's implementation, rather than the details of the result or proof.

We begin with a standard statement of Vickrey's theorem and proof, in this case from \citet{mas-04}:
\begin{theorem}[Vickrey 1961]
  In a second-price auction, it is (weakly) dominant for each buyer $i$ to bid its valuation $v_i$. Furthermore, the auction is efficient.
\end{theorem}

\begin{proof}[Proof \#1]
  Suppose that buyer $i$ bids $b_i < v_i$. The only circumstance in which the outcome for $i$ is changed by its bidding $b_i$ rather than $v_i$ is when the highest bid $b$ by other bidders satisfies $v_i > b > b_i$. In that event, buyer $i$ loses by bidding $b_i$ (for which its net payoff is $0$) but wins by bidding $v_i$ (for which its net payoff is $v_i - b$). Thus, it is \textit{worse} off bidding $b_i < v_i$. By symmetric argument, it can only be worse off bidding $b_i > v_i$.  We conclude that bidding its valuation (truthful bidding) is weakly dominant. Because it is optimal for buyers to bid truthfully and the high bidder wins, the second-price auction is efficient.
\end{proof}

However intelligible to humans, \citeauthor{mas-04}'s proof is too stylized for computers: that there is only one circumstance in which changing bids changes the outcome is merely asserted; the ``symmetric argument'' is not explicitly elaborated.  Before formalizing it, we therefore elaborated the paper proof, and restructured it to four cases, rather than the original nine:

\begin{proof}[Proof \#2]
  Let $N$ be the set of bidders, and suppose bidder $i$ bids $b_i = v_i$, whatever $b_j$ each other bidder $j \ne i$ bids.  There are two cases:
  \begin{enumerate}[nosep]
    \item $i$ wins.  This implies $b_i  = v_i = \max_{j \in N} \left\{ b_j \right\}$,
    $p_i = \max_{j \in N \backslash \left\{ i \right\}} \left\{ b_j \right\}$, and
    $u_i \left( \bm{b} \right) = v_i - p_i \ge 0$.  Now consider $i$ submitting an arbitrary bid $\hat{b}_i \ne b_i$ so that the bid vector is $\left( b_1, \ldots, b_{i-1}, \hat{b}_i, b_{i+1}, \ldots, b_n \right)$.  This has two sub-cases:
    \begin{enumerate}[nosep]
      \item $i$ wins with $\hat{b}_i$, so that $u_i \left( b_1, \ldots, b_{i-1}, \hat{b}_i, b_{i+1}, \ldots, b_n \right) = u_i \left( \bm{b} \right)$: $i$ receives the same utility from winning the item, and pays the same price as the second highest bid has not changed.
      \item $i$ loses with $\hat{b}_i$, so that $u_i \left( b_1, \ldots, b_{i-1}, \hat{b}_i, b_{i+1}, \ldots, b_n \right) = 0 \le u_i \left( \bm{b} \right)$.
    \end{enumerate}
    \item $i$ loses. This implies $p_i = 0$, $u_i \left( \bm{b} \right) = 0$, and
      $b_i \le \max_{j \in N \backslash \left\{ i \right\}} \left\{ b_j \right\}$ as, otherwise, $i$ would have won. This yields again two cases for $i$'s alternative bid $\hat{b}_i$:
    \begin{enumerate}[nosep]
      \item $i$ wins, so that
        $u_i \left( b_1, \ldots, b_{i-1}, \hat{b}_i, b_{i+1}, \ldots, b_n \right) = v_i - \max_{j \in N \backslash \left\{ i \right\}} \left\{ b_j \right\} = b_i - \max_{j \in N \backslash \left\{ i \right\}} \left\{ b_j \right\} \le 0 = u_i \left( \bm{b} \right)$.
      \item \label{case:l-l} $i$ loses, so that $u_i \left( b_1, \ldots, b_{i-1}, \hat{b}_i, b_{i+1}, \ldots, b_n \right) = 0 = u_i \left( \bm{b} \right)$.
    \end{enumerate}
  \end{enumerate}
  By analogy for all $i$, $\bm{b} = \bm{v}$ supports an equilibrium in weakly dominant strategies.  Efficiency is immediate: the highest bidder has the highest valuation.
\end{proof}

To formally prove Vickrey's theorem, we used Isabelle, whose higher-order logic allows our formalization to remain close to paper mathematics.

Our proof, \textit{Vickrey.thy}, is a 9 KB, 185 line file that draws on five ancillary files written for this project.\footnote{See \url{https://github.com/formare/auctions/tree/master/isabelle/Auction} for the code.}  All six files amount to 17 KB and 404 lines – much longer than their paper counterparts.  A more reliable estimate of the additional effort involved in formal proofs, the \textit{de Bruijn factor}~\citep{Wiedijk:tdbf:on}, cleans and compresses files before dividing the size of the code by the size of an informal {\TeX} source.  It thus avoids bias by semantically irrelevant differences in the syntaxes of formalisations such as languages or code styles using different lengths of lines or of identifiers.  The de Bruijn factor relating Proof \#2 and its definitions (including $\max$) to our Isabelle code is 1.1; as our \TeX\ source is more elaborate than usual, this is lower than the typically observed factors of around four.

Figure \ref{fig:th_gr} depicts the files used in the proof.  Those already in Isabelle's library are marked by ellipses.  Dotted ellipses denote files containing general definitions and lemmas that we have added to Isabelle's library.  Rectangles denote this paper's auction-specific files.  Directed edges denote dependence, with the source code being imported into the target code.

\begin{figure}[ht]
  \begin{tikzpicture}
    \input{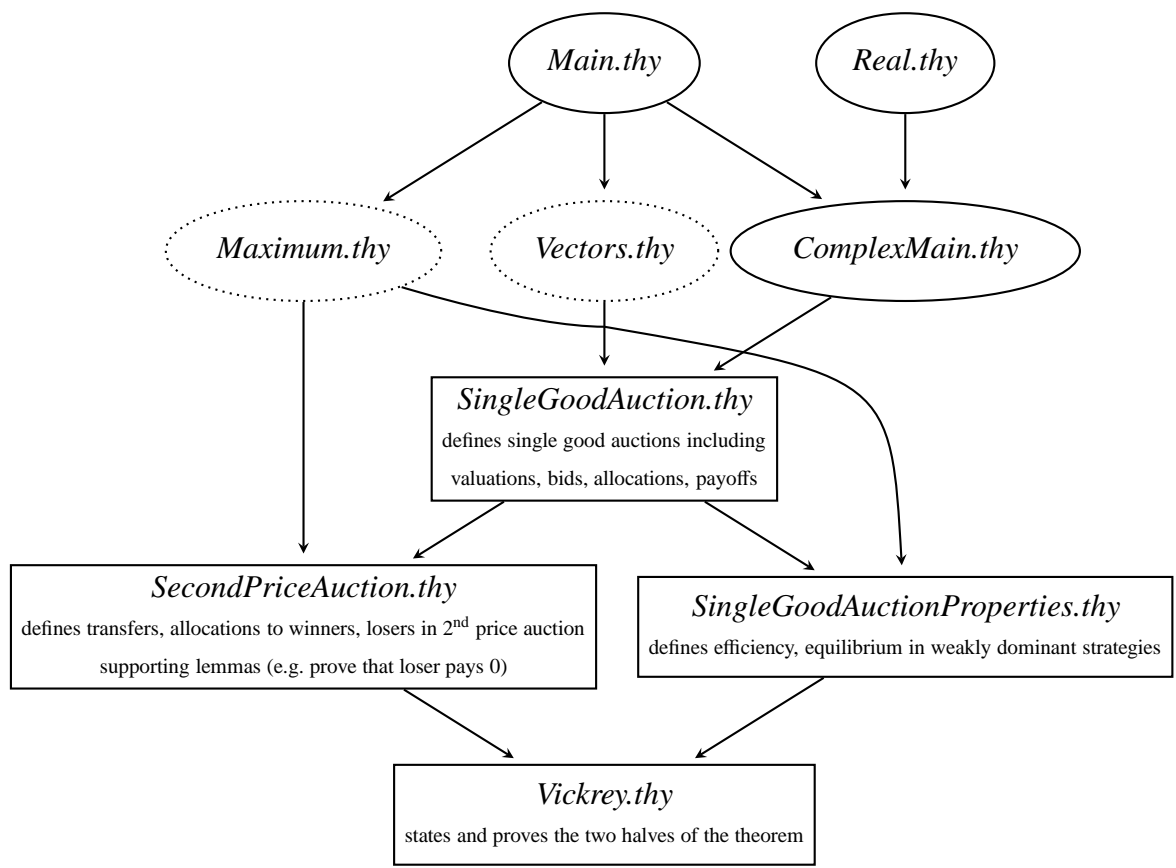}
  \end{tikzpicture}
  \caption{High level theory graph for the formal proof of Vickrey's theorem}
  \label{fig:th_gr}
\end{figure}

\textit{Vickrey.thy} begins with \textit{vickreyA}, which proves that truth telling is weakly dominant in Vickrey auctions:

\begin{pmi}
  \pmikw{theorem} \pmiid*{vickreyA}:\pminl
  \quad \pmikw{fixes} N::\pmitm{\pmiid{participant} \pmiid*{set}}   \pmikw{and} v::\pmiid{valuations} \pmikw{and} A::\pmiid*{single\_good\_auction}\pminl
  \quad \pmikw{assumes} \pmiid*{val}: \pmitm*{\pmiid{valuations} N\ v} \pminl
  \quad \pmikw{defines} \pmitm*{b ≡ v}\pminl
  \quad \pmikw{assumes} \pmiid*{spa}: \pmitm{\pmiid{second\_price\_auction} A}\pmikw{and} \pmiid*{card\_N}: \pmitm*{\pmiid{card} N > 1}\pminl
  \quad \pmikw{shows} \pmitm*{\pmiid{equilibrium\_weakly\_dominant\_strategy} N\ v\ b\ A}\pmill
\end{pmi}

The \textbf{fixes} keyword applies the theorem to any $N$, $v$ and $A$ of the given types.  The type \textit{single\_good\_auction} is defined as an $\mathit{input}\times \mathit{output}$ relation, with the bidders and their bids as input, and a Boolean allocation vector and a vector of transfers as outcome.\footnote{This can be seen from expressions such as $((N, b'), (x', p')) \in A$.}  The $\mathit{valuations}$ type is defined elsewhere to be a vector of real numbers.
The \textbf{assumes} keyword on the next line states that the theorem holds under an assumption labeled \textit{val}, namely that in the vector $\mathit{v}$ of $\mathit{N}$ real numbers, all numbers are non-negative (this defined at another place as the definition of `valuations'). \todo{CR@MK: is this usage of \textit{valuations} a function, a predicate?}

Next, the \textbf{defines} declaration equates bids and valuations.  The following \textbf{assumes} keyword introduces and labels further assumptions (e.g.\ $A$ is a second-price auction; $N$ contains more than one bidder).  The \textbf{shows} keyword states the theorem: $N$ agents participating in auction $A$, with valuations $v$ and bids $b$ (equated to valuations) yields an equilibrium in weakly dominant strategies.

\textit{SingleGoodAuctionProperties.thy} defines the equilibrium concept:
\begin{pmi}
\pmikw{definition} \pmiid{equilibrium\_weakly\_dominant\_strategy} ::\pminl
\quad  \pmitm{\pmiid{participant} \pmiid*{set} ⇒ \pmiid*{valuations} ⇒ \pmiid*{bids} ⇒ \pmiid*{single\_good\_auction} ⇒ \pmiid*{bool}} \pmikw*{where}\pminl
\quad  \pmitm*{\pmiid{equilibrium\_weakly\_dominant\_strategy} N\ v\ b\ A\ ⟷\pminl
\qquad    \pmiid{valuations} N\ v ∧ \pmiid{bids} N\ b ∧ \pmiid{single\_good\_auction} A\ ∧\pminl
\qquad    (∀i ∈ N\,.\pminl
\qquad\quad      (∀\pmiid*{whatever\_bid}\,.\,\pmiid{bids} N\ \pmiid*{whatever\_bid} ⟶\pminl
\qquad\qquad        (\pmikw{let} b' = \pmiid*{whatever\_bid}(i := b\ i)\pminl
\qquad\qquad\quad         \pmikw{in} (∀ x\ p\ x'\ p'\,.\,((N, \pmiid*{whatever\_bid}), (x, p)) ∈ A ∧ ((N, b'), (x', p')) ∈ A\pminl
\qquad\qquad\qquad          ⟶ \pmiid{payoff}\ (v\ i)\ (x'\ i)\ (p'\ i) ≥ \pmiid{payoff}\ (v\ i)\ (x\ i)\ (p\ i)))))}\pmill
\end{pmi}

The definition's second line declares the type of the \textit{equilibrium\_weakly\_dom}\allowbreak \textit{inant\_strategy} 
to be a (Boolean) predicate whose arguments are a set of participants, a valuation vector, a bid vector, and an auction.\footnote{The $A⇒B⇒C$ notation, referred to as \textit{currying}, is equivalent to $A\times B\to C$, but is conceptually simpler as it does not require definition of a $\times$ operation.}  The definition's body states that the predicate, given arguments $N$, $v$, $b$ and $A$, evaluates to true if and only if the remaining expression does.  The expressions in the subsequent line ensure that all arguments have admissible values.
\todo{CR@CL: I've cut this explanation. I'm aiming at minimal explanations - giving readers enough of a sense of the objects, without turning this into a tome.}
Similarly, our first step when introducing $\mathit{whatever\_bid}$ is to ensure that it is an admissible bid vector.  The $\mathit{whatever\_bid}(i := b\ i)$ notation then takes an arbitrary vector and replaces its $i$th component with $i$'s bid $b\ i$ (which the theorem equates to $i$'s valuation).\footnote{The code snippet contains various instances of ``.'': these are separators that improve readability.}

We denote the outcome of an arbitrary bid ($\mathit{whatever\_bid}$) by $\left( \bm{x}, \bm{p} \right)$, while $\left( \bm{x}', \bm{p}' \right)$ denotes that of $i$'s original bid and arbitrary bids by agents $j \ne i$.  To satisfy the definition of an equilibrium in weakly dominant strategies, the outcome $\left( \bm{x}', \bm{p}' \right)$ of $i$'s truthful bid must yield a payoff no less than that resulting from an arbitrary bid. The $\pmikw{let} \cdots\ \pmikw{in} \cdots$ notation\footnote{We use ``$\cdots$'' to distinguish the standard use of ellipses from Isabelle's ``\dots'' notation, whose meaning we introduce when explaining line 30 of the following code snippet.} introduces local abbreviations, which can only be accessed within the $\pmikw*{in}$ block; here, this makes the expression $((N, b'), (x', p')) ∈ A$ more readable. \todo{CR@CL: I don't understand what the alternative to the local variables is; how would it be less readable?  I'm also wondering whether this misses the main point: we use the ``let \ldots in'' notation to show something, not just to manipulate local variables.\\
CL@CR: No, here we are not using the “let \ldots in” notation to \emph{show} anything, but really just as syntactic sugar.  Example: “let k = f(x,g(y)) in h(z,k)” is the same as “h(z,f(x,g(y)))”.  It really pays off once $k$ occurs more than once in the “in” part, but sometimes it also makes a single occurrence more readable.  In the {\LaTeX} source below this comment I'll put a variant of the source lines without the “let”.  Choose which one you like better.  Note that in one footnote above we are making a reference to $b'$, so this footnote would then have to be rephrased.\\
MK@CR: Let's speak instead of ``local variable'' of an ``abbreviation''. We can this way use $b'$ as a shorthand notation for $whatever_bid(i := b i)$.}


The code snippet below formalizes case \ref{case:l-l} of Proof \#2.  It is \textit{declarative}, resembling a textbook proof.  \textit{Procedural} proofs, by contrast, prescribe \textit{tactics} to apply, thus more resembling the \textit{process} humans use to find proofs.  In either case, each theorem creates a \textit{proof obligation}, or a \textit{goal}; these may be broken into \textit{subgoals} (e.g.\ by case distinction); the set of local proof obligations implied by these subgoals are stored on a \textit{goal stack}.

\begin{proof}[Proof \#3]\mbox{}\
\begin{lstlisting}[mathescape,numbers=left]
$\pmikw{proof} -$
$\quad (* \cdots *)$
$\quad \{$
$\qquad \pmikw{fix} i::\pmiid*{participant}$
$\qquad \pmikw{assume} \pmiid{i\_range}: \pmitm*{i ∈ N}$
$\qquad (* \cdots *)$
$\qquad \pmikw{let} ?b = \pmitm*{\pmiid*{whatever\_bid}(i := b\ i)}$
$\qquad (* \cdots *)$
$\qquad \pmikw{have} \pmiid*{weak\_dominance}: \pmitm*{\pmiid{payoff} (v\ i)\ (x'\ i)\ (p'\ i) ≥ \pmiid{payoff} (v\ i)\ (x\ i)\ (p\ i)}$
$\qquad \pmikw{proof} \pmiid*{cases}\label{line:proof-cases1}$
$\qquad\quad \pmikw{assume} \pmiid*{non\_alloc}: \pmitm*{x'\ i ≠ 1}$
$\qquad\quad \pmikw{with} \pmiid{spa\_pred'} \pmiid{i\_range} \pmikw{have} \pmitm{x'\ i = 0} \pmikw{using} \pmiid{spa\_allocates\_binary} \pmikw{by} \pmiid*{blast}\label{line:blast1}$
$\qquad\quad \pmikw{with} \pmiid{spa\_pred'} \pmiid{i\_range} \pmikw{have} \pmiid*{loser\_payoff}: \pmitm*{\pmiid{payoff} (v\ i)\ (x'\ i)\ (p'\ i) = 0}$
$\qquad\qquad \pmikw{by} (\pmiid{rule} \pmiid*{second\_price\_auction\_loser\_payoff})\label{line:by-rule1}$
$\qquad\quad \pmikw{have} \pmiid*{i\_bid\_at\_most\_second}: \pmitm*{?b\ i ≤ \pmiid*{?b\_max'}}$
$\qquad\quad \pmikw{proof} (\pmiid{rule} \pmiid*{ccontr})\label{line:proof-rule}$
$\qquad\qquad \pmikw{assume} \pmitm*{¬ \pmiid*{?thesis}}\label{line:assume1}$
$\qquad\qquad \pmikw{then} \pmikw{have} \pmitm{?b\ i > \pmiid*{?b\_max'}} \pmikw{by} \pmiid*{simp}\label{line:then-have1}\label{line:by-simp1}$
$\qquad\qquad \pmikw{with} \pmiid{defined} \pmiid{spa\_pred'} \pmiid{i\_range} \pmikw{have} \pmitm*{\pmiid{second\_price\_auction\_winner} N\ ?b\ x'\ p'\ i}\label{line:with1}$
$\qquad\qquad\quad \pmikw{by} (\pmiid{simp} \pmiid*{add}:\ \pmiid*{only\_max\_bidder\_wins})\label{line:by-simp2}$
$\qquad\qquad \pmikw{with} \pmiid{non\_alloc} \pmikw{show} \pmiid{False}\label{line:with2}$
$\qquad\qquad\quad \pmikw{unfolding} \pmiid{second\_price\_auction\_winner\_def}\label{line:unfolding-start}$
$\qquad\qquad\qquad \pmiid{second\_price\_auction\_winner\_outcome\_def} \pmikw{by} \pmiid*{blast}\label{line:unfolding-end}\label{line:blast2}$
$\qquad\quad \pmikw{qed}$
$\qquad\quad \pmikw{show} \pmiid*{?thesis}$
$\qquad\quad \pmikw{proof} \pmiid*{cases}\label{line:proof-cases2}$
$\qquad\qquad \pmikw{assume} \pmitm*{x\ i ≠ 1}\label{line:assume2}$
$\qquad\qquad \pmikw{then} \pmikw{have} \pmitm{x\ i = 0} \pmikw{by} (\pmiid{rule} \pmiid*{spa\_allocates\_binary'})\label{line:xi0}\label{line:then-have2}\label{line:by-rule2}$
$\qquad\qquad \pmikw{with} \pmiid{spa\_pred} \pmiid{i\_range} \pmikw{have} \pmitm{\pmiid{payoff} (v\ i)\ (x\ i)\ (p\ i) = 0}\label{line:have}$
$\qquad\qquad\quad \pmikw{by} (\pmiid{rule} \pmiid*{second\_price\_auction\_loser\_payoff})\label{line:by-rule3}$
$\qquad\qquad \pmikw{also} \pmikw{have} \pmitm{... = \pmiid{payoff} (v\ i)\ (x'\ i)\ (p'\ i)} \pmikw{using} \pmiid{loser\_payoff} ..\label{line:also-have}\label{line:..}$
$\qquad\qquad \pmikw{finally} \pmikw{show} \pmiid{?thesis} \pmikw{by} (\pmiid{rule} \pmiid*{eq\_refl})\label{line:finally}\label{line:by-rule4}$
$\qquad\quad \pmikw{next}$
$\qquad\qquad (* \cdots *)$
$\qquad\quad \pmikw{qed}$
$\qquad \pmikw{next}$
$\qquad\quad (* \cdots *)$
$\qquad \pmikw{qed}$
$\quad \}$
$\quad (* \cdots *)$
$\pmikw{qed}$
\end{lstlisting}

\end{proof}

The \textbf{proof} keyword starts the proof.  Invoked alone, Isabelle would automatically select inference rules to apply.  \textbf{proof –} performs manual inference.  Alternatively, one can specify existing inference rules:
\begin{itemize}
  \item \textbf{proof} \textit{cases} (lines \ref{line:proof-cases1} and \ref{line:proof-cases2}) makes a case distinction; analysis of each case concludes by \textbf{show}ing that the desired thesis holds; \textbf{qed} clears the goal stack; \textbf{next} begins the next case.
  \item \textbf{proof} \textit{(rule ccontr)} (line \ref{line:proof-rule}) undertakes proof by contradiction, culminating in \textbf{show} \textit{False}.
\end{itemize}

The proof considers an arbitrary but fixed participant $i$, which is introduced locally with the $\pmikw*{fix}$ keyword, and assumed to be in the admissible range $N$ for bidders.\footnote{In Isabelle, the descriptive form of a verb (e.g.\ \textbf{fixes}, \textbf{assumes} or \textbf{shows}) are often used when stating theorems, while their imperative counterparts (e.g.\ \textbf{fix}, \textbf{assume} or \textbf{show}) are used locally in proofs.}

The \textbf{have} statements establish local facts, generating local proof obligations, which have to be discharged by corresponding \textbf{proof}s.  Here, the \textit{cases} proof establishes that $u_i \left( \cdots, v_i, \cdots \right) \ge u_i \left( \cdots, b_i, \cdots \right)$.  This proof makes use of further facts, omitted to keep the snippet readable: $\pmiid*{spa\_pred}$ and $\pmiid*{spa\_pred'}$ state that $((N, \mathit{whatever\_bid}), (x, p))$ and 
$((N, ?b), (x', p'))$ respectively are in an $(\mathit{input}, \mathit{outcome})$ relationship of a second price auction with each other.\footnote{Isabelle syntactically substitutes identifiers starting with \textit{?} by other, usually more complex expressions before checking a proof step.  Syntactic substitution is performed, for example, by the preprocessor of many programming languages, allowing the programmer to use shorthand designations rather than writing complicated expressions in full.  It is distinct from the semantic equation of two variables, as in ``$b \equiv v$''.}  $\pmiid*{defined}$ states that a vector with one component per element of the (finite) set $N$ has a well-defined maximum component.

Both \textbf{from} and \textbf{using} introduce facts to discharge the \textbf{have} obligations.  The \textbf{by} keyword invokes an automated proof method, instead of discharging proof obligations by explicit declarative means.  Isabelle thus combines ATP and ITP methods.
\begin{enumerate}
  \item \textit{simp} (lines \ref{line:by-simp1} and \ref{line:by-simp2}) simplifies (e.g.\ $x \wedge x = x$) the statement to be proved.  Line \ref{line:by-simp2} supplies a simplification rule of our own, $\pmiid{only\_max\_bidder\_wins}$.
  \item \textit{blast} (lines \ref{line:blast1} and \ref{line:blast2}) ``is (in principle) a complete proof procedure for first-order formulas''~\citep{isabelle-prog-prove}.  In practice, \textit{blast} either succeeds, fails, or -- giving a practical example of semi-decidability -- runs until the user cancels it.
  \item\label{it:rule} \textit{rule} (lines \ref{line:by-rule1}, \ref{line:proof-rule}, \ref{line:by-rule2}, \ref{line:by-rule3} and \ref{line:by-rule4}) applies the given lemma as an inference rule.
      In line~\ref{line:..}, ``$..$''  abbreviates $\pmikw{by} \pmiid*{rule}$, which automatically applies a matching inference rule.
\end{enumerate}
While interactively developing the proof, we employed the \textbf{try} and \textbf{try0} commands, which apply a range of automated methods, to find the most appropriate proof methods.  Automated calls can always be replaced by explicit declarative steps; Isabelle's Sledgehammer tool~\citep{isabelle-sledgehammer} can sometimes provide them automatically.

The \textbf{assume} $\cdots$ \textbf{then have} constructions (lines \ref{line:assume1} and \ref{line:then-have1}, and \ref{line:assume2} and \ref{line:then-have2}) list assumptions \textbf{then} state the proof obligations.  Line \ref{line:assume1}'s identifier \textit{?thesis} refers to the proof obligation at the proof's current level of reasoning.

Lines \ref{line:unfolding-start} -- \ref{line:unfolding-end}'s \textbf{unfolding} also performs substitutions, replacing stated concepts' names with the bodies of their definitions.  Unlike abbreviations with \textit{?}, the latter are semantic definitions, of which the reasoner make use (e.g.\ \textit{second\_price\_auction\_winner\_def} is restated in terms of $i \in N$, $i \in \arg\max \bm{b}$, \ldots).

Lines \ref{line:have}–\ref{line:finally}'s \textbf{have} $\cdots$ \textbf{also have} $\cdots$ \textbf{finally show} construction allows chains of reasoning with equality before discharging a proof obligation: the ``\dots'' following the \textbf{also have} are replaced by the right hand side of the previous \textbf{have} statement.  In line \ref{line:also-have}, this establishes that $i$ receives zero given valuation $v_i$ and either $\left( \bm{x}, \bm{p} \right)$, or $\left( \bm{x}', \bm{p}' \right)$.

\section{Discussion} \label{se:concle}

\todo{CL@CR: Not that this is necessarily a problem, but this discussion only covers the state of the art as estabished by \emph{others} (it is thus, once more, a literature review), but it does not refer to any of \emph{our} work (i.e.\ Section~\ref{se:mpv}).  CR@CL: I think this is OK; this paragraph is v.\ brief, so just a highlight; I don't think that our work is highlightable yet.}The decade since the mechanized reasoning community became interested in economic applications has seen rapid progress.  When \citeauthor{nip-09} reported on his formalization of Arrow's theorem, he agreed that ``[s]ocial choice theory turns out to be perfectly suitable for mechanical theorem proving'', but felt that it was ``unclear if [it] will lead to new insights into either social choice theory or theorem proving'' \citep{nip-09}.  However, that very year \citet{ta-li-09} used mechanized reasoning to discover a new theorem that subsumes Arrow's, which \citet{ch-se-14} believed to be novel, and unlikely to have been found with traditional methods.  Shortly thereafter, \citet{ge-en-11} contributed their 84 impossibility theorems.

If mechanized reasoning is to make further inroads into economics it must be sensitive to a number of concerns.  First, economics has no proofs of comparable complexity or length to significant results in modern mathematics.  Thus, the question of whether a proof will exceed the capability of human theorists to verify is less of a concern than in mathematics.  Further, it is unclear that there have been any disastrous cases of mistaken proofs within economics; instead, our greater errors likely result from poor modelling in the first place, and coding or data errors in econometrics. 

Second, even when mechanized reasoners have helped identify new results, economic theorists may dismiss them as unmotivated, non-transparent or lacking insight.\footnote{See \citet[p.73]{av-ha-14} for a discussion of the tension between rigour and insight in pure mathematics.}  Even, however, in the worst case, we believe that a stock of poorly-motivated, non-transparent theorems generated blindly by computer provide cases for us to think about and reason with: the presence of the intermediate independence axiom in all of the larger impossibility theorems found by \citet{ge-en-11} should provide precisely the sort of hunch that sets us sharpening our pencils.

We close by suggesting some further possible applications of mechanized reasoning to economic problems.

First, there are open problems in auction theory that seem amenable to solution by computation (rather than `reasoning').  For example, the simplest formulation of optimal multi-object auctions \citep[q.v.][]{arm-00} defines a linear programming problem that quickly becomes too large to solve manually as the number of items increases.\footnote{See \citet{ar-ro-99} for the equivalent multi-dimensional screening problem for a monopolist.}  As efficient algorithms exist for solving linear programming problems, \textit{automated mechanism design} \citep[q.v.][]{co-sa-03} has already begun to address the purely computational aspects of optimal mechanism design.  As formal methods can be used to verify the results of computations \citep[q.v.][]{gon-08,hales2015formal}, proofs in automated mechanism design could also be verified by formal methods. \todo{CR@all: [no action required] Crampton, Klemperer not aware of anything subsequent to/more authoritative than \citet{arm-00}. \citet{co-sa-03} seem to be able to handle special cases with discrete valuation distributions.  n.b.\ \citet{arm-00} also has discrete valuations. Where do they properly describe their \S5.2.2 `experiment'?}

Second, we believe that the exhaust-then-induct technique pioneered by \citet{ta-li-09}, and developed by \citet{ge-en-11}, offers the promise of automating search for theorems in other areas of economic theory.  The formal similarities between social choice and matching theory -- including a reliance on discrete objects -- suggests that this technique could be applied directly to the latter.  Although auction theory appears richer in its use of continuous objects (prices), there is a small literature establishing results by induction \citep{ch-se-07,mo-se-15,ada-14,ka-oh-ta-15}; the possibility of coupling their induction steps with computational exhaustion has not been explored.


However these tools are applied within economics, it is hard to imagine them not becoming more important, as the tools themselves become faster and easier to use, as they gain acceptance within the pure mathematics community, and as the mechanized reasoning community seeks more applications for them.




\printbibliography

\end{document}
